%% file: main.tex
\documentclass[a4paper,onecolumn,accepted=2025-03-15]{quantumarticle}
\pdfoutput=1
\input{macros}

\def\anon{0}

\title{On the Computational Hardness of Quantum One-Wayness
}

\ifnum\anon=0
\author{
Bruno P. Cavalar}

\email{bruno.cavalar@cs.ox.ac.uk}
\affiliation{
    Department of Computer Science, University of Oxford,
    Oxford,
    UK
}
\author{Eli Goldin}
\affiliation{
    Department of Computer Science, New York University, 
    New York,
    USA
}
\email{eli.goldin@nyu.edu}
\author{Matthew Gray}
\affiliation{
    Department of Computer Science, University of Oxford,
    Oxford,
    UK
}
\email{matthew.gray@magd.ox.ac.uk}
\author{Peter Hall}
\affiliation{
    Department of Computer Science, New York University, 
    New York,
    USA
}
\email{pf2184@nyu.edu}
\author{Yanyi Liu}
\affiliation{
    Department of Computer Science, Cornell Tech,
    USA
}
\email{yl2866@cornell.edu}
\author{Angelos Pelecanos}
\affiliation{
    Department of Computer Science,
    UC Berkeley,
    USA
}
\email{apelecan@berkeley.edu}
\fi

\begin{document}

\maketitle

\input{abstract}

\setcounter{tocdepth}{2}
\tableofcontents
  
\input{introduction}
\input{techoverview}

\input{conclusion}

\input{relatedwork}

\input{prelims}

\input{prs-owsg}

\input{break-owsg-pp-oracle}

\section*{Acknowledgements}

This work was developed when B. P. Cavalar was a PhD student at the
University of Warwick.
B. P. Cavalar acknowledges support from
the Chancellor's International Scholarship of the University of Warwick
and the Royal Society University Research Fellowship
URF\textbackslash R1\textbackslash211106.
A. Pelecanos was supported by DARPA under Agreement No. HR00112020023.
E. Goldin was supported by a National Science Foundation Graduate Research Fellowship.
This work was done in part while the authors were visiting 
the Simons Institute for the Theory of Computing.

\bibliographystyle{quantum}	
\bibliography{main}

\end{document}

%% file: macros.tex
\usepackage{amsmath,amssymb,amsthm}
\usepackage{algorithm}
\usepackage{algorithmic}
\usepackage{subcaption}
\usepackage{caption}
\usepackage[table]{xcolor}
\usepackage{xcolor}
\usepackage{tcolorbox}
\usepackage{tikz}
\usepackage{bbm}
\usepackage{wasysym}
\usepackage{geometry}
\usepackage{verbatim}
\usepackage{graphicx}
\usepackage{tikz}
\usepackage{amsfonts}
\usepackage{hyperref}
\usepackage{mathtools}
\usepackage{csquotes}
\usepackage{listings}
\usepackage{thmtools}
\usepackage{cleveref}
\usepackage{tikz-cd}
\usepackage{minibox}
\usepackage{braket}
\usepackage{soul}
\usepackage{arydshln}
\usepackage[numbers]{natbib}

\makeatletter
\newcommand\ketbra[2][]{%
  \def\ketbra@firstarg{#1}%
  \def\ketbra@secondarg{#2}%
  \ifx\ketbra@firstarg\empty%
    \left\lvert\ketbra@secondarg\middle\rangle \! \middle\langle \ketbra@secondarg\right\rvert%
  \else%
    \left\lvert\ketbra@firstarg\middle\rangle \! \middle\langle\ketbra@secondarg\right\rvert%
  \fi%
}
\makeatother

\makeatletter
\def\metadef#1#2{%
  \def\metadef@iter##1{\ifx##1;\else \expandafter\newcommand\csname#1\endcsname{#2}\expandafter\metadef@iter\fi}%
  \expandafter\metadef@iter%
}
\makeatother

\metadef{vec#1}{\mathbf #1}abcdefghijklmnopqrstuvwxyz;

\metadef{ten#1}{\mathbf #1}ABCDEFGHIJKLMNOPQRSTUVWXYZ;

\metadef{mat#1}{\mathbf #1}ABCDEFGHIJKLMNOPQRSTUVWXYZ;

\metadef{frak#1}{\mathfrak #1}ABCDEFGHIJKLMNOPQRSTUVWXYZ;

\metadef{c#1}{\mathcal #1}ABCDEFGHIJKLMNOPQRSTUVWXYZ;

\metadef{s#1}{\mathsf #1}ABCDEFGHIJKLMNOPQRSTUVWXYZ;

\newtheorem{theorem}              {Theorem}
\numberwithin{theorem}{section}
\newtheorem{claim}       [theorem] {Claim}         
\newtheorem{lemma}       [theorem] {Lemma}
\newtheorem{corollary}   [theorem] {Corollary}

\newtheorem{proposition} [theorem] {Proposition}

\newtheorem{definition}  [theorem] {Definition}

\newtheorem*{theorem*}{Theorem}

\newcommand{\eps}{\varepsilon}

\newcommand{\poly}{\mathsf{poly}}
\newcommand{\negl}{\mathsf{negl}}

\newcommand{\E}{\mathop{\mathbb{E}}}

\newcommand{\abs}[1]{\left|#1\right|}

\newcommand{\tld}[1]{\widetilde{#1}}

\renewcommand{\set}[1]{\left\{#1\right\}}

\def\sseq{\subseteq}

\newcommand{\cala}{{\mathcal A}}
\newcommand{\calb}{{\mathcal B}}

\newcommand{\cald}{{\mathcal D}}

\newcommand{\calk}{{\mathcal K}}
\newcommand{\call}{{\mathcal L}}

\newcommand{\setbrac}[1]{\left\{#1\right\}}
\newcommand{\blt}{\setbrac{0,1}}

\newcommand{\bbn}{{\mathbb N}}

\newcommand{\bbs}{{\mathbb S}}

\makeatletter
\newcommand\ExpOwsg[4][]{%
  \def\ExpOwsg@firstarg{#1}%
  \ifx\ExpOwsg@firstarg\empty%
    \mathsf{Exp}_{#2,#3}(#4)
  \else%
    \mathsf{Exp}_{#2,#3,#1}(#4)
  \fi%
}
\makeatother

\newcommand{\CNOT}{\mathsf{CNOT}}

\def\Exp{\mathop{{}\mathbb{E}}}
\renewcommand{\P}{\mathsf{P}}
\newcommand{\NP}{\mathsf{NP}}
\newcommand{\BQP}{\mathsf{BQP}}
\newcommand{\PP}{\mathsf{PP}}
\newcommand{\PostBQP}{\mathsf{PostBQP}}
\newcommand{\PromPostBQP}{\mathsf{PromisePostBQP}}
\newcommand{\PromPP}{\mathsf{PromisePP}}
\newcommand{\Ver}{\mathsf{Ver}}
\newcommand{\Samp}{\mathsf{Samp}}

\newcommand{\A}{\mathcal{A}}
\newcommand{\K}{\mathcal{K}}
\newcommand{\Haar}{\text{\normalfont Haar}}
\newcommand{\HaarN}[1]{\text{\normalfont Haar}(#1)}
\newcommand{\Dkeycond}[1]{(\Samp \, |\, #1 )^{\mathsf{key}}}
\newcommand{\Sampx}[1]{\mathcal{S}^{\leq #1}}

\newcommand{\probabilityfont}[1]{\mathbf{#1}}

\newcommand{\rndY}{\probabilityfont{Y}}

\newcommand{\rndb}{\probabilityfont{b}}
\newcommand{\rndB}{\probabilityfont{B}}
\newcommand{\flws}{\sim}

\newcommand{\dTV}{\mathsf{d}_{\mathsf{TV}}}

%% file: abstract.tex
\begin{abstract}

There is a large body of work studying what forms of computational hardness are needed to realize classical cryptography. In particular, one-way functions and pseudorandom generators can be built from each other, and thus require equivalent computational assumptions in order to be realized. Furthermore, the existence of either of these primitives implies that $\P \neq \NP$, which gives a lower bound on the necessary hardness.

One can also define versions of each of these primitives with quantum output: respectively one-way state generators and pseudorandom state generators. Unlike in the classical setting, it is not known whether either primitive can be built from the other. Although it has been shown that pseudorandom state generators for certain parameter regimes can be used to build one-way state generators, the implication has not been previously known in full generality. Furthermore, to the best of our knowledge the existence of one-way state generators has no known implications in traditional complexity theory.

We show that pseudorandom states compressing $n$ bits to 
$\log n + 1$
qubits 
can be used to build
one-way state generators and that pseudorandom states compressing $n$ bits to $\omega(\log n)$ qubits are \emph{themselves} one-way state generators.
This is a nearly optimal result, since pseudorandom states with fewer than $c \log n$-qubit output can be shown to exist unconditionally. We also show that any one-way state generator can be broken by a quantum algorithm with classical access to a $\PP$ oracle.

An interesting implication of our results is that a 
$t(n)$-copy
one-way state generator exists unconditionally, for every
$t(n) = o(n/\log n)$. This contrasts nicely with
the previously known fact that
$O(n)$-copy one-way state generators require computational hardness. 
We also outline a new route towards a black-box separation between one-way state generators and quantum bit commitments.

\end{abstract}

%% file: introduction.tex
\section{Introduction}
The vast majority of useful classical cryptographic primitives share the following property: they can be used to build one-way functions in a black-box
manner. 
A one-way function is a function on bit-strings which can be
efficiently evaluated but is hard to invert. In this sense, one-way
functions can be thought of as a ``minimal'' cryptographic primitive.
However, any one-way function can be broken by an efficient algorithm with
access to an $\NP$ oracle. This means that if $\P = \NP$, then one-way
functions do not exist. As it is unknown whether $\P = \NP$ or not, the existence
of one-way functions, and thus all of classical cryptography, must rely on computational assumptions.

This issue led to the natural desire to ``map out'' the world of classical cryptography. Over many years, cryptographers have done a fairly good job of figuring out which cryptographic primitives can be built from each other. This cartography helps give a sense of the relative strength of assuming the existence of different cryptographic primitives. 

As an example, it is known how to construct one-way functions from any key
exchange protocol, i.e. a protocol where two parties can agree on a secret
using only communication over a public channel ~\cite{BCG89}. However,
there is strong evidence that building a key exchange protocol from a
one-way function is difficult~\cite{IR89}. The primitives which can be
built from one-way function form a crypto-complexity class known as
``MiniCrypt"~\cite{Impagliazzo95}. Two cryptographic primitives in this
class of particular note are pseudorandom generators and commitment
schemes. A pseudorandom generator is a deterministic function which maps a
small amount of randomness to a longer string indistinguishable from
random. A commitment scheme is a process by which a party can encode some
string into a ``commitment'', such that later the party can prove this
``commitment'' was an encoding of the original string.

In recent years, cryptographers have started to consider what happens if we allow cryptographic primitives to have quantum output. Here, the landscape of relations between primitives looks very different. Of particular note, it was shown that quantum key distribution, a quantum variant of key exchange, exists unconditionally~\cite{BB84,Wiesner83}. On the other hand, it is known that the quantum versions of one-way functions, pseudorandom generators, and commitments cannot be secure against information-theoretic attackers, and thus require some computational hardness in order to exist~\cite{LC97,ji2018pseudorandom,dakshita23commitments}. These variants are known as one-way state generators, pseudorandom state generators, and quantum bit commitments respectively. 

However, it is still unclear what this hardness looks like from a complexity perspective. In particular, it is known there exists an oracle relative to which $\BQP \supseteq \NP$, but all three of these primitives exist\footnote{For one-way state generators and quantum bit commitments, the result follows from~\cite{Kretschmer21Quantum} and the subsequent works of~\cite{morimae2022quantum, dakshita23commitments}.}~\cite{Kretschmer21Quantum}. Furthermore, we are still mapping out the relations between quantum primitives. It was only recently discovered that quantum bit commitments can be built from one-way state generators~\cite{dakshita23commitments}, and it is still an open question as to whether pseudorandom state generators can be built from quantum bit commitments.

The main goal of this work is to broaden our understanding of the hardness of quantum primitives, with a particular focus on one-way state generators. In particular, we show two main results:
\begin{enumerate}
    \item One-way state generators can be built from pseudorandom states
        for nearly all parameter regimes requiring computational hardness.
    \item If one-way state generators exist, then $\BQP \neq \PP$.
\end{enumerate}

These main results bring along a number of interesting implications. The following are of particular note:
\begin{enumerate}
    \item A fixed-copy version of one-way state generators exists unconditionally.
    \item If we can show that quantum bit commitments exist relative to a $\PP$ oracle, then there is a black-box separation showing it is unlikely that we will be able to build one-way state generators from quantum bit commitments.
\end{enumerate}

\subsection{Results}

We now recall a few key concepts from quantum cryptography
and give more details about the results we show.

\paragraph{Pseudorandom State Generators (PRS).} A pseudorandom state generator, originally defined in~\cite{ji2018pseudorandom}, is a quantum variant of a pseudorandom generator. Given a classical key $k$, a PRS maps $k$ to a quantum pure state $\ket{\phi_k}$. The security guarantee is that the output of a PRS on a random input should look like a random state. That is, it is hard for any quantum adversary to distinguish any polynomial number of copies of a random $\ket{\phi_k}$ from polynomial copies of a Haar random state. 

The relationship between the length of the input key $n$ and the number of output qubits $m$ determines whether a PRS can exist information-theoretically or requires computational assumptions. In particular, ~\cite{ananth2022pseudorandom} shows that PRSs with output state length $m \geq \log n$ qubits can be broken by an inefficient adversary, and thus must be a computational object. On the other hand, it is known that PRSs with state length $m \leq c\log n$ exist unconditionally for some $c \in (0,1)$~\cite{ananth2022pseudorandom,brakerski2020scalable}.

\paragraph{One Way State Generators (OWSG).} 
A one-way state generator, originally defined in~\cite{morimae2022one}, is
a quantum variant of a one-way function. Just like for PRS, a OWSG maps
a classical key $k$ to quantum state $\ket{\phi_k}$. The security guarantee
of a OWSG is that, given any polynomial number of copies of $\ket{\phi_k}$,
it is hard for a quantum algorithm to find keys $k'$ such that $\ket{\phi_k}, \ket{\phi_{k'}}$ have noticeable overlap. OWSGs can also be
defined to have mixed state outputs~\cite{morimae2022quantum}, although we
will not consider this variant in this work.

\paragraph{Building OWSGs from PRSs.}
It is known that any expanding PRS is also a OWSG~\cite{morimae2022one}. 
Here, an expanding PRS is one which has keys of length $n$ and output
states of length $m > (1+c)n$ for 
some $c > 0$.
We extend this proof to show
that any PRS with output length at least 
$m \geq \log n + 1$ implies OWSGs.
Since OWSGs require computational hardness \cite{dakshita23commitments,
LC97}, and there exists $d < 1$ such that PRSs with output length $\leq d
\log n$ exist unconditionally \cite{brakerski2022computational}, this
reduction is close to optimal.

\begin{theorem}[Informal version of Theorem~\ref{thm:prs-imply-weak-owsg}]
    \label{t:prs-owsg-optimal}
    For every $c \geq 1$, 
    if there exists a PRS
    mapping
    $n$-bit
    strings 
    to $(\log n+c)$-qubit states,
    then OWSGs exist.
\end{theorem}

Through a closely related argument,
we also find that PRSs that map $n$ bits to $\omega(\log n)$ qubits are
OWSGs.

\begin{theorem}[Informal version of Theorem~\ref{thm:prs-imply-strong-owsg}]
\label{thm:prs-owsg-intro}
     Any PRS that maps $n$-bit strings to
     $\omega(\log n)$-qubit 
     states is also a OWSG.
\end{theorem}

\paragraph{Fixed-copy PRSs and OWSGs.}

Both PRSs and OWSGs are defined to be secure against adversaries that are given any polynomial number of copies of the output state. However, we could instead consider an alternative definition 
where we fix the number of copies given to the adversary. We will refer to these primitives by the names $t$-copy PRS and $t$-copy OWSG. Related primitives have already been considered in a number of works, including~\cite{gottesman2001quantum,dakshita23commitments,LMW23}.
A summary of prior work and our results can be found in \Cref{tab:prior-result-summary}.

It is known that for any fixed function $t$, expanding $t(\lambda)$-copy
PRSs require computational hardness because they can be used to construct
quantum bit 
commitments~\cite{LMW23,yan2022general,morimae2022one,brakerski2022computational},
where $\lambda$ is the security parameter.
Therefore, any expanding $t(\lambda)$-copy PRS can be broken by an
inefficient attacker. On the other hand, if we do not have an expansion
requirement, it can be shown that something called an efficient approximate
$t$-design (defined formally in Section~\ref{ssec:t-design}) is also a
$t$-copy PRS ~\cite{Kretschmer21Quantum}. Since efficient approximate
$t$-designs exist
unconditionally~\cite{dankert2009exact,haferkamp2023efficient,odonnell2023explicit},
so do $t$-copy (non-expanding) PRSs.

Thus, one may ask the question: for what parameters do $t$-copy OWSGs require
computational assumptions? In a recent work, Khurana and Tomer \cite{dakshita23commitments} show that Weisner
encodings / BB84 states~\cite{Wiesner83, BB84} are $1$-copy OWSGs. Additionally, written twenty years before OWSGs were defined,~\cite{gottesman2001quantum}
shows that $t$-qubit stabilizer states are $t/2$-copy OWSGs. 
The OWSG construction of~\cite{gottesman2001quantum} only has a weaker security guarantee, 
but this can be resolved by amplification.

Note that this means that for any fixed polynomial $t$, $t(\lambda)$-copy
OWSGs exist unconditionally,
where $\lambda$ is the security parameter.
However,~\cite{dakshita23commitments} shows that quantum bit commitments
can be built from $\Theta(n)$-copy OWSGs, where $n$ is the input key length.
This is not a contradiction, since the number of copies of security here
depends on the input length instead of the security parameter. Thus, we may
consider the following refinement of our question:

\begin{quote}
\emph{For what functions
$t(\cdot)$ do $t(n)$-copy OWSGs require computational assumptions?}
\end{quote}

Our proof of
\Cref{thm:prs-owsg-intro}
will also imply the following result.
\begin{corollary}[Informal version of Corollary~\ref{cor:t-design}]\label{cor:t-design-intro}
    Every efficient approximate $t$-design 
    mapping $n$ bits to 
    $\omega(\log n)$
    qubits
    is also a $(t-1)$-copy OWSG.
\end{corollary}

\begin{table}[]
    \centering
    \begin{tabular}{c|c|c|c}
        Primitive & Copies & Security & Comments \\
        \hline
        Expanding PRS & $t \geq 1$ copy & Computational & \cite{morimae2022one, yan2022general, brakerski2022computational} \\
        PRS & $\poly(\lambda)$-copy & Statistical & Approximate $t$-designs \\
        OWSG & $\poly(\lambda)$-copy & Statistical & \Cref{thm:prs-imply-strong-owsg} with approximate $t$-designs \\[0.5em]
        \hdashline\noalign{\vskip 0.5em} 
        PRS & $O(n/\log n)$-copy & Statistical & Approximate $t$-designs \\
        OWSG & $O(\sqrt{n})$-copy & Statistical & Stabilizer states \cite{gottesman2001quantum} \\
        OWSG & $\Omega(n)$-copy & Computational & Implies quantum bit commitments \cite{dakshita23commitments} \\
        OWSG & $o(n/\log n)$-copy & Statistical & \Cref{cor:unconditional-owsg} 
    \end{tabular}
    \caption{A summary of what is known about the computational and
    information-theoretic nature of quantum cryptographic primitives, based
on the number of copies of the output given to the adversary. We say that
security is computational if the existence of the primitive requires
computational hardness and we say that the security is statistical if the
primitive can be shown to exist unconditionally against statistical
adversaries.
The rows above the dashed line correspond to constructions where the number of copies is in terms of the security parameter $\lambda$.
The rows below the dashed line correspond to constructions where the number of copies is in terms of $n$, the number of input bits. 
}
    \label{tab:prior-result-summary}
\end{table}

If we consider state-of-the-art constructions of approximate $t$-designs~\cite{odonnell2023explicit}, we in addition prove the following.

\begin{corollary}[Equivalent to Corollary~\ref{cor:unconditional-owsg}]\label{cor:unconditional-owsg-intro}
    For every $t(n) = o(n/\log n)$, there exists a $t(n)$-copy OWSG.
\end{corollary}

Aside from demonstrating an interesting new property of approximate
$t$-designs, these results give an interesting dichotomy: OWSGs require
computational hardness for $\Omega(n)$-copies, and exist unconditionally for
$o(n/\log n)$ copies.

\paragraph{Quantum cryptography, computational complexity, and
separations.} It is known
from~\cite{Kretschmer21Quantum,ananth2022pseudorandom} that the existence
of a PRS which outputs a 
$(\log n + O(1))$-qubit state implies that $\BQP \neq \PP$,
where $\BQP$ refers to the class of problems efficiently solvable by
quantum computers, and $\PP$ refers to the class of problems such that a
probabilistic Turing machine gets the correct answer with probability
strictly greater than $\frac{1}{2}$.  Thus, like with one-way functions,
the existence of PRSs has implications in complexity theory.

However, no similar results are known about OWSGs or quantum bit
commitments.
In fact, it is conjectured
by~\cite{LMW23} that quantum bit commitments may exist relative to a random
oracle and \emph{any} classical oracle, even ones that depend on the random
oracle. Khurana and Tomer~\cite{dakshita23commitments} observe that there exists a classical
oracle that breaks OWSGs, which implies that such a conjecture cannot
extend to the existence of OWSGs. If the conjecture of \cite{LMW23} is
proven, this would provide a black box separation between OWSGs and quantum
bit commitments.
However, it is not immediately clear that the oracle~\cite{dakshita23commitments} 
mentions lies inside any interesting complexity class.

We show that the existence of OWSGs indeed does have interesting complexity
implications. In particular, we show the following.
\begin{theorem}[Informal version of
    Corollary~\ref{cor:pp-break-owsg}]\label{thm:pp-break-owsg-intro} If
OWSGs exist, then $\BQP \neq \PP$.  \end{theorem}

It then follows that a black box separation between quantum bit commitments
and OWSGs can be achieved by proving a weaker version of the conjecture
of~\cite{LMW23}, namely that there exists an oracle $\mathcal{O}$ relative
to which quantum bit commitments exist and $\PP^\mathcal{O} \sseq
\BQP^\mathcal{O}$.

%% file: techoverview.tex
\subsection{Technical Overview}

We now give an overview of the main technical ingredients of our work.

\subsubsection{Building one-way state generators from pseudorandom state generators}

\paragraph{Expanding PRS.} We will begin this section by detailing the argument from~\cite{morimae2022one} that an expanding PRS is also a OWSG. Recall that an expanding PRS maps $n$-bit strings to $m$-qubit quantum states, where $m \geq (1+c)n$ for some $c > 0$. The natural reduction which uses the OWSG adversary to also break the PRS works. Let $\A$ be a OWSG adversary which outputs $k'$ such that $\ket{\phi_{k'}}$ close to $\ket{\phi_k}$ using $t$ copies of $\ket{\phi_k}$. Given $t+1$ copies of a state $\ket{\psi}$, we can test whether it is an output of the PRS or Haar-random as follows: run $\A$ on the first $t$ copies to get a state $\ket{\phi_{k'}}$, and compare $\ket{\phi_{k'}}$ with the last copy of $\ket{\psi}$. 

If $\ket{\psi} = \ket{\phi_k}$ for some $k$, then $\ket{\phi_{k'}}$ will be close to $\ket{\phi_k}$. If $\ket{\psi}$ is Haar-random, then since it is a random state on $(1+c)n$ qubits, with high probability it is far from $\ket{\phi_k}$ for all $2^n$ values of $k$. This is because $\frac{2^n}{2^{(1+c)n}}$ is negligible in $n$.

\paragraph{Improvement to shrinking PRS.} To improve this result to PRS with
$O(\log n)$ bit output, we simply improve the analysis of the exact same
reduction. We make the following simple observation about Haar-random
states: for any fixed $m$-qubit state $\ket{\phi}$, the probability that a
Haar-random state $\ket{\psi}$ is ``close'' to $\ket{\phi}$ is
$2^{-\Omega(2^m)}$. Thus, in the reduction above, if $\ket{\psi}$ is a
Haar-random state on $\omega(\log n)$ qubits, 
then,
by a union bound, with high probability it is far from $\ket{\phi_k}$ for
all $2^n$ values of $k$. This argument is enough to derive
Theorem~\ref{thm:prs-owsg-intro}.
Using
amplification~\cite{morimae2022quantum}, we also find that OWSGs can be
built
even from a PRS which outputs $\log n + O(1)$ qubits,
obtaining
\Cref{t:prs-owsg-optimal}.

\paragraph{Building OWSGs from a fixed-copy PRS.}
We can instantiate our reduction with a $t$-copy PRS (i.e. a PRS that is secure against $t$ copies). Our reduction then shows that any $t$-copy PRS is also a $(t-1)$-copy OWSG. The fact that approximate $t$-designs are $t$-copy PRSs~\cite{Kretschmer21Quantum} then gives Corollary~\ref{cor:t-design-intro}.

Recently, O'Donnell, Servedio, and Paredes \cite{odonnell2023explicit} showed that there exists an efficient
$2^{-\lambda}$-approximate $t$-design on $m$-qubit quantum states with 
seed length $n = O(mt + \lambda)$. 
Setting $m = \omega(\log n)$ and $t
= o\left(\frac{n}{\log n}\right)$ shows that approximate
$o\left(\frac{n}{\log n}\right)$-designs 
with $\omega(\log n)$-output bits exist,
and thus 
$o\left(\frac{n}{\log n}\right)$-copy OWSGs also
exist. That is,
Corollary~\ref{cor:unconditional-owsg-intro} holds. 

\subsubsection{Breaking one-way state generators with a $\PP$ oracle}

\paragraph{The power of $\PP$ in the quantum setting.} $\PP$ is typically defined by referring to Turing machines or randomized algorithms. These definitions are not very useful when dealing with quantum computing. However, it turns out that $\PP$ has an equivalent formulation with much more obvious quantum applications, known as $\PostBQP$ ~\cite{Aaronson2004QuantumCP}. $\PostBQP$ refers to the class of problems efficiently solvable by uniform quantum circuits with the additional ability to \textit{postselect}. Postselection is another word for performing conditional sampling, and in the quantum setting refers to the ability to choose the result of a measurement and acquire the corresponding residual state. Thus, to break any OWSG with a $\PP$ oracle, it suffices to define an algorithm which breaks the OWSG given oracle access to some language in $\PostBQP$.

\paragraph{One-way puzzles.} Instead of breaking OWSGs with a $\PostBQP$ oracle directly, we rely on a recent result which constructs an interesting classical output primitive, a one-way puzzle, from any OWSG~\cite{dakshita23commitments}. Formally, a one-way puzzle is a pair of algorithms $(\Samp, \Ver)$ where $\Samp$ samples a key-puzzle pair $(k,s)$ such that $\Ver(k,s)$ outputs $1$ with overwhelming probability. $\Samp$ is required to be an efficient quantum algorithm, and $\Ver$ is allowed to be any arbitrary function. The security requirement is that given $s$, it is hard for an adversary to find a $k'$ such that $\Ver(k',s)=1$. Morally, a one-way puzzle is a one-way function where the input and output are sampled together.

It is clear that any one-way puzzle $(\Samp,\Ver)$ can be broken by an
adversary with the ability to postselect. Given a puzzle $s$, the adversary
can simply run $\Samp$ up until it would measure the output state, and then
postselect on the output puzzle being $s$. Measuring the output key will
then give a $k'$ which with high probability will satisfy $\Ver(s,k')=1$.
This attack can be viewed as a search version of the attack of~\cite{Kretschmer21Quantum} 
against any $\lambda$-output PRS.

\paragraph{Attack with decision oracle.} But note that this attack makes use of postselecting directly. It is unclear how to translate this into an attack with a decision oracle for $\PP$. To solve this, we show that it is possible in general to perform conditional sampling given access to a $\PP$ oracle:
\begin{lemma}[Informal version of Lemma~\ref{lem:sample-bit-by-bit}]\label{lem:sample-bit-by-bit-intro}
    Let $\Samp$ be a (uniform) quantum polynomial time algorithm 
    such that
    $\Samp(1^n)$ outputs a
    pair of classical strings $(k, s)$.
    There exists a poly-time
    quantum algorithm $\mathcal{A}$ and a $\PP$ language $\mathcal{L}$ such
    that $\mathcal{A}^\mathcal{L}$ takes as input $s'$ and outputs $k'$, 
    and whose distribution 
    has total variation distance at most $1/n$ from the distribution $\Dkeycond{s'}$ defined by
    $$\Pr[\Dkeycond{s'} \to k'] = \Pr_{\Samp(1^n) \to (k,s)}[k = k' | s = s'].$$
    (In other words,
    we denote by $\Dkeycond{s'}$
    the distribution of keys generated by $\Samp(1^n)$
    conditioned on the puzzle being equal to $s'$.)
\end{lemma}

This lemma is in some sense a ``search-to-decision'' style argument for
$\PP$. The argument goes along the same lines as the search-to-decision
reduction for $\mathsf{SAT}$. Note that with a $\PostBQP$ oracle, we can test
whether or not it is possible for an algorithm $\Samp$ to produce any given
output $x$ by simply postselecting on $x$ being produced by $\Samp$. 

A naive approach to sampling $k'$ from $\Dkeycond{s'}$ is to sample $k'$ bit by bit. It is possible with a $\PP$ oracle to check whether any given output is in the range of $\Samp$. Thus, we can begin by checking whether $(1, s')$ is in the range of $\Samp$ with all but the first bit of the key discarded. If so, we can set the first bit of $k'$ to be $1$, otherwise $0$. In the next step, we can check whether $(k'_1\circ 1, s')$ is in the range of $\Samp$ with all but the second bit of the key discarded. If so, we can set the second bit of $k'$ to be $1$, otherwise $0$. Repeating this process for each bit of the key will uniformly select an output $k'$ from the range of $\Dkeycond{s'}$.

However, $\Dkeycond{s'}$ is not necessarily a flat distribution, and so the resulting distribution on $k'$ may be very different from $\Dkeycond{s'}$. However, this can be resolved by noticing a key fact. Using a $\PP$ oracle, it is possible to \textit{estimate the probability} that the first bit of the output of $\Samp$ is $1$ conditioned on the output puzzle being $s'$. Thus, we can use the same technique as before, but instead of just setting each bit of $k'$, we can sample each bit according to our approximation of the correct conditional distribution. Although the error will add up, we can set our initial error to be small enough that the distribution over $k'$ will be sufficiently close to $\Dkeycond{s'}$.

Applying Lemma~\ref{lem:sample-bit-by-bit-intro} to our postselecting attack against one-way puzzles gives us an efficient quantum attack against one-way puzzles using a $\PP$ oracle. As one-way puzzles can be built from one-way state generators, this immediately implies Theorem~\ref{thm:pp-break-owsg-intro}.

\paragraph{More on search-to-decision reductions using a $\PP$ oracle.} The
ability of a postselection oracle to aid in search-to-decision reductions
was first noted by~\cite{INNRY22}, where it was shown that a quantum
poly-time algorithm can find a $\mathsf{QMA}$ witness by making one quantum
query to a $\PP$ oracle. They use very different techniques, and in fact
our algorithm requires many classical queries instead of one quantum query.

%% file: conclusion.tex
\subsection{Conclusion and Future work}

We showed two new results about one-way state generators.
We proved that their existence is implied by almost all computational PRSs, and
that $O(n)$-copy OWSGs
can be broken by QPT algorithms with access to a $\PP$
oracle. 
These results bring the cartography of OWSGs much more in line
with that of PRSs. 

\paragraph{Further research directions.}
We outline two major lines of research
which are left open by this paper.
\\

\noindent\emph{Separating OWSG and quantum bit commitments.} 
The two recent papers \cite{dakshita23commitments} and \cite{LMW23} make
clear that proving that quantum bit commitments exist relative to a random
oracle and any classical oracle (with access to the random oracle) would
suffice to show a black box separation between OWSGs and quantum bit
commitments. It follows from our results 
that a black box separation between quantum bit commitments
and OWSGs can be achieved by proving a weaker 
statement, namely that there exists an oracle $\mathcal{O}$ relative
to which quantum bit commitments exist and $\PP^\mathcal{O} \sseq
\BQP^\mathcal{O}$.

Since it is now known from \cite{dakshita23commitments} that 
OWSGs imply quantum bit commitments, proving this conjecture and the accompanying
separation would provide strong evidence that quantum bit commitments are
(at least among the quantum cryptographic primitives so far proposed) the
minimal assumption for quantum cryptography.\\

\noindent\emph{OWSG $\rightarrow$ PRS.}
We have now shown that 
PRSs imply OWSGs
in nearly the greatest possible generality,
and that the existence of either have equivalent known
complexity ramifications. This gives additional motivation to the question
of whether these primitives are equivalent (as their classical equivalents are known to be).
Proving either an equivalence or separation (even for a limited
copy setting) would be an exciting accomplishment that would greatly
clarify the cartography of quantum cryptography.\\

\paragraph{Open questions.}
There are two open questions which we would like to see
an answer for.
\\

\noindent\emph{Optimal PRS $\rightarrow$ OWSG Reduction.}
OWSGs are inherently computational objects, meaning they cannot be implied by statistical primitives. Currently we know that 
\begin{itemize}
    \item PRSs with $m \geq \log n$ output qubits require computational assumptions.
    \item There exists a $c \in (0,1)$ s.t. statistical PRSs with $m \leq c\log n$ output qubits exist. So the existence of PRSs in this regime cannot imply OWSGs.
    \item PRS with $m \geq \log n + 1$ output qubits imply OWSGs.
\end{itemize} 
This leaves the question open of whether a PRS with $m \in (c \log n,\, \log n]$ output qubits imply the existence of OWSGs. The smallest improvement to this space would be to show that a PRS with $m = \log n$ output qubits imply OWSGs. 
Other possible tightening would be to show that, for any $c < 1$, a statistical PRS with $m \leq c\log n$ exist, though this would show the surprising result that just over 1 bit of randomness
per amplitude is enough to achieve statistical closeness to Haar-random
states.\\

\noindent\emph{Parameter regimes for which OWSGs don't exist.} 
It would be very interesting to know which parameter regimes it can be shown OWSGs do not exist in. We conjecture that much
like how $O(\log n)$ output classical OWFs cannot exist, $m = O(\log\log n)$ output qubit OWSGs cannot exist unconditionally.
It would be good to have a complete
understanding of how 
OWSGs
behave
in different parameter
regimes, and a proof of its non-existence in that setting would help
to clarify this picture.\\

%% file: relatedwork.tex
\subsection{Concurrent and further work}

We remark that a concurrent revision of~\cite{morimae2022one} provides a
different proof of \Cref{t:prs-owsg-optimal}. In particular, in Appendix C
they show that, given a PRS $G: k\mapsto \ket{\phi_k}$ mapping $n$ bits to
$m \geq \log n$ qubits, the mapping $k \mapsto \ket{\phi_k}^{\otimes n}$ is
a OWSG. Note, however, that these techniques do not easily extend to imply
our other results. In particular, our results show how to build compressing
OWSGs (as in \Cref{thm:prs-owsg-intro}) and thus also give better
parameters for unconditional $t$-copy OWSGs (as in
\Cref{cor:t-design-intro}).

After our results were announced, Hhan, Morimae and
Yamakawa~\cite{DBLP:journals/corr/abs-2312-16025}
proved that there does not exist OWSGs with $O(\log n)$-output qubits. This shows
that our proof that PRSs with $\omega(\log n)$-output qubits 
are OWSGs (\Cref{thm:prs-owsg-intro}) is
optimal under the cryptographic assumption that (post-quantum) one-way functions
exist, since one-way functions imply PRSs of arbitary output
length~\cite{brakerski2020scalable}.
A later work of Batra and Jain~\cite{10756159} then proved that our unconditional
construction of $o(n / \log n)$-copy OWSGs is optimal.
More recently, Khurana and Tomer~\cite{cryptoeprint:2024/1490} announced a
series of results, including a proof that one-way puzzles imply the
hardness of probability estimation. Their result is inspired by
the arguments developed in \Cref{sec:break-pp}.

Recent works have also demonstrated an oracle separation between
OWSGs and one-way
puzzles~\cite{behera2025newworlddepthsmicrocrypt,
bostanci2024oracleseparationquantumcommitments}.
As noted in \Cref{sec:break-pp},
our $\PP$ oracle breaks both OWSGs and one-way puzzles.

%% file: prelims.tex
\section{Preliminaries}

In this section, we introduce basic notation and recall definitions
from the literature that will be used throughout the rest of this work.

\subsection{Computational complexity}

We refer the reader to~\cite{DBLP:books/daglib/0023084}
for the definition of standard complexity classes such as
$\BQP$ and $\PP$.
We will often abreviate quantum polynomial-time algorithms
by the acronym `QPT'.

We recall the definition of the $\PostBQP$ complexity class
(see, e.g.,~\cite[Definition 12]{Kretschmer21Quantum}).
We will only be concerned with its promise version in this work.

\begin{definition}[\textsf{PromisePostBQP}]
    A promise problem $\Pi:\{0, 1\}^* \to \{0, 1, \bot\}$ is in
    $\PromPostBQP$ (Postselected Bounded-error Quantum Polynomial time) if
    there exists a 
    QPT
    algorithm
    $\cala$
    whose output is in $\set{0,1,*}$,
    and which is such that:
    \begin{enumerate}
        \item[(i)] For all $x \in \blt^*$,
            we have $\Pr[\mathcal{A}(x) \in \{0, 1\}] > 0$. 
            When $\mathcal{A}(x) \in \{0, 1\}$, we say that \emph{postselection succeeds}.
        \item[(ii)] If $\Pi(x) = 1$, then $\Pr[\mathcal{A}(x) = 1 \mid \mathcal{A}(x) \in \{0, 1\}] \geq \frac{2}{3}$. In other words, conditioned on postselection succeeding, $\mathcal{A}$ outputs $1$ with at least $\frac{2}{3}$ probability.
        \item[(iii)] If $\Pi(x) = 0$, then $\Pr[\mathcal{A}(x) = 0 \mid
            \mathcal{A}(x) \in \{0, 1\}] \geq \frac{2}{3}$. In other words,
            conditioned on postselection succeeding, $\mathcal{A}$ outputs
            $1$ with at most $\frac{1}{3}$ probability.
    \end{enumerate}
\end{definition}

We remark that the definition of $\PostBQP$ is sensitive to the choice of
the gate-set, as noticed by~\cite{v011a006}.
We remark that~\cite{v011a006}
also proves that every gate-set that satisfies two
``reasonable'' conditions gives rise to an equivalent
``canonical'' $\PostBQP$ class.
Throughout the paper we assume that we are dealing with one such
gate-set, such as the $\set{\CNOT, H, T}$-gate-set.
We refer the reader to~\cite[Section 2.5]{v011a006}
for a detailed technical discussion of this matter.

We also recall the following result of
Aaronson~\cite{Aaronson2004QuantumCP},
which states that $\PP = \PostBQP$,
remarking that his result also holds 
for the corresponding promise classes.
\begin{lemma}[Aaronson \cite{Aaronson2004QuantumCP}]
    \label{lem:postbqp-eq-pp}
    $\PromPostBQP = \PromPP$.
\end{lemma}
Since $\PP$ is a syntactic class,
we can extend any 
promise problem in $\PromPP$
into a language in $\PP$.
This remark will be fundamental in~\Cref{sec:break-pp},
as it was also in~\cite{Kretschmer21Quantum},
when we build a $\PP$ oracle with which we can break
OWSGs.

\subsection{Quantum information theory and cryptography}

We denote by $\bbs(N)$ the set of
$N$-dimensional pure quantum states.
We identify the set of $n$-qubit pure states with
$\bbs(2^n)$.
We denote by
$\HaarN{n}$
the Haar measure on $\bbs(2^n)$.

\paragraph{Quantum cryptography.}

We start by introducing the notion of pseudorandom state generators (PRS)~\cite{ji2018pseudorandom}.
Roughly speaking, given a classical key $k$, a PRS maps $k$ to a quantum pure state $\ket{\phi_k}$. The security guarantee is that the output of a PRS on a random input should look like a random state. That is, it is hard for any quantum adversary to distinguish a random $\ket{\phi_k}$ from a Haar random state. 

\begin{definition}[Pseudorandom States, Definition 2 of \cite{ji2018pseudorandom}]
    Let $\lambda$ be the security parameter and 
    let $\calk$ be a set of binary strings referred to as the \emph{key space}.
    Let $G$ be a QPT algorithm that on input $k \in \calk$ outputs a pure quantum
    state $\ket{\phi_k}$ over $n = n(\lambda)$ qubits.
    We say $G$ is $(t, \eps)$-pseudorandom if the distribution over
    outputs is $\eps$-indistinguishable from Haar random given
    $t$ copies.
    In other words, for any QPT adversary $\A$,
    we have
    $$\abs{\Pr_{k\leftarrow \K}[\A(\ket{\phi_k}^{\otimes t}) = 1] -
    \Pr_{\ket{\psi}\leftarrow\HaarN{n}}[\A(\ket{\psi}^{\otimes t}) = 1]} \leq
    \eps.$$

    We say that $G$ is pseudorandom if it is $\left(\lambda^c,
    \frac{1}{\lambda^c}\right)$-pseudorandom for all $c > 0$, and that it is t-pseudorandom or a t-copy PRS if it is $\left(t,
    \frac{1}{\lambda^c}\right)$-pseudorandom for all $c > 0$.
\end{definition}

We turn to introducing the notion of one-way state generators (OWSGs)~\cite{morimae2022one}. 
A OWSG maps a classical key $k$ to quantum state $\ket{\phi_k}$. The security guarantee
of a OWSG is that, given any polynomial number of copies of $\ket{\phi_k}$,
it is hard for a quantum algorithm to find keys $k'$ such 
$\ket{\phi_k}, \ket{\phi_{k'}}$ have noticeable overlap. 
OWSGs can also be
defined to have mixed state outputs~\cite{morimae2022quantum}, although we
will not consider this variant in this work. 

\begin{definition}[One-Way State Generators, Definition 4.1 of \cite{morimae2022quantum}]
    \label{def:owsg}
    Let $\lambda$ be the security parameter
    and
    let $\calk$ be a set of binary strings referred to as the \emph{key space}.
    Let $G$ be a QPT algorithm that on input 
    $k \in \calk$ 
    outputs a pure quantum state $\ket{\phi_k}$.
    We say
    G is $(t,\eps)$-one-way if the outputs are hard to invert with
    accuracy at least $\eps$.
    In other words, for any QPT adversary $\A$,
    we have
    $$\Exp_{
    \substack{
        k \leftarrow \K
        \\
        k' \leftarrow \A(\ket{\phi_k}^{\otimes t})
    }}
    \left[\abs{\braket{\phi_k|\phi_{k'}}}^2 \right] \leq \eps.$$
    We say that $G$ is one-way or strongly one way if it is $\left(\lambda^c, \frac{1}{\lambda^c}\right)$-one-way for all $c > 0$, and that it is t-one-way, a t-copy OWSG, or a t-copy strong OWSG if it is $\left(t,
    \frac{1}{\lambda^c}\right)$-one-way for all $c > 0$.
\end{definition}

We remark that all of the results in this paper
will also hold for 
the more general definition of (pure-state) one-way state generators 
that was introduced
in the later work of Morimae and
Yamakawa~\cite{morimae2022one}, where the one-way state generator is
allowed to have a separate quantum key generation 
procedure.\footnote{The results we show in
\Cref{sec:owsg-prs} imply the existence of one-way state generators in
the sense of \Cref{def:owsg} in various settings. 
Since these are more restricted objects than the ones considered in the
more general definition of~\cite{morimae2022one}, our results will also implies their
existence in the same settings and with an analogous parameter regime.
Furthermore, since our oracle QPT algorithm in \Cref{sec:break-pp} comes
from an algorithm breaking \emph{one-way puzzles}, and one-way puzzles were
proved in~\cite{dakshita23commitments} to follow from the ``general'' pure state OWSGs of~\cite{morimae2022one},
our algorithm will therefore also break 
OWSGs with a quantum key generation procedure.
}
Our techniques, however, do not generalize to 
the mixed-state OWSGs defined in the later work of Morimae
and Yamakawa~\cite{morimae2022one}.

We will rely on the notion of one-way puzzles, defined in~\cite{dakshita23commitments}. A one-way puzzle is a pair of algorithms $(\Samp, \Ver)$ where $\Samp$ samples a key-puzzle pair $(k,s)$ such that $\Ver(k,s)$ outputs $1$ with overwhelming probability. $\Samp$ is required to be an efficient quantum algorithm, and $\Ver$ is allowed to be any arbitrary function. The security requirement is that given $s$, it is hard for an adversary to find a $k'$ such that $\Ver(k',s)=1$.

\begin{definition}[One-Way Puzzles \cite{dakshita23commitments}]
    \label{def:owpuz}
    Let $\lambda$ be the security parameter.
    A one-way puzzle is a pair of sampling and verification algorithms $(\Samp, \Ver)$ with the following syntax.

    \begin{itemize}
        \item $\Samp(1^\lambda) \to (k, s)$ is a (uniform) quantum
            polynomial time algorithm that outputs a pair of classical
            strings $(k, s)$. We refer to $s$ as the puzzle and $k$ as its
            key. Without loss of generality we may assume that $k \in \{0,
            1\}^\lambda$.
        \item $\Ver(k, s) \to \top~\text{or}~\bot$, is an unbounded
            algorithm that on input any pair of classical strings $(k, s)$
            halts and outputs either $\top$ or $\bot$.
    \end{itemize}
    These satisfy the following properties.
    \begin{itemize}
        \item \textbf{Correctness.} Outputs of the sampler pass verification with overwhelming probability, i.e.,
        $$\Pr_{(k, s) \leftarrow \Samp(1^\lambda)}[\Ver(k, s) = \top] = 1 - \negl(\lambda).$$
        \item \textbf{Security.} Given $s$, it is (quantumly) computationally infeasible to find $k$ satisfying $\Ver(k, s) = \top$, i.e., for every polynomial-sized quantum circuit $\mathcal{A}$,
        $$\Pr_{(k, s) \leftarrow \Samp(1^\lambda)}[\Ver(\mathcal{A}(s), s) = \top] = \negl(\lambda).$$
    \end{itemize}
\end{definition}

An important recent result due to Khurana and Tomer~\cite{dakshita23commitments}
states that one-way state generators imply one-way puzzles.
This will be crucial for our algorithm in \Cref{sec:break-pp}.

\begin{theorem}[\protect{\cite[Theorem 4.2]{dakshita23commitments}}]
    \label{thm:owsg-to-owpuz}
    If there exists a $(O(n), \negl(n))$-OWSG, then there exists a one-way puzzle.
\end{theorem}

\subsection{Probability distributions}

We recall a few standard probability distributions, the chi-squared distribution
and the Fisher-Snedecor distribution,
the latter of which will be useful when trying to show that
compressing PRSs are one-way state generators~\Cref{lem:goodlem}.

\begin{definition}[Chi-squared distribution]
    \label{def:chi-square}
    Let $X_1,\dots,X_k$
    be independent normal random variables
    with mean $0$ and variance $1$.
    The 
    \emph{chi-square distribution with $k$ degrees of freedom},
    denoted by $\chi^2(k)$,
    is the probability distribution of the sum
    $\sum_{i=1}^k X_i^2$.
\end{definition}

\begin{definition}[Fisher-Snedecor distribution]
    \label{def:fs-dstr}
    Let $a,b \in \bbn$.
    The 
    \emph{Fisher-Snedecor distribution with $a$ and $b$ degrees of freedom},
    denoted by $F(a, b)$,
    is the distribution given by the 
    following ratio:
    \begin{equation*}
        F(a,b)
        \flws
        \frac{\chi^2(a)/a}{\chi^2(b)/b}
        =
        \frac{b}{a}
        \cdot
        \frac{\chi^2(a)}{\chi^2(b)}.
    \end{equation*}
\end{definition}

We now recall the definition of the beta function and its variations
in order to express the cumulative distribution function
of the $F$-distribution.
The 
\emph{beta function} is defined as
\begin{equation*}
    B(a,b)
    :=
    \int_0^1
    x^{a-1}
    (1-x)^{b-1}
    \,
    dx.
\end{equation*}
The
\emph{incomplete beta function}
is defined as
\begin{equation*}
    B(t; a,b)
    :=
    \int_0^t
    x^{a-1}
    (1-x)^{b-1}
    \,
    dx.
\end{equation*}
The \emph{regularised incomplete beta function}
is defined as
\begin{equation*}
    I_x(a,b)
    :=
    \frac{B(x; a,b)}{B(a,b)}.
\end{equation*}

\begin{lemma}[Cumulative distribution function of the $F$-distribution]
    \label{lem:cumulative-function-fs}
    The cumulative distribution function $p_{a,b}(t)$ 
    of 
    the
    $F(a,b)$ distribution satisfies
    $$
        p_{a,b}(t) = I_{at/(at + b)}(a/2,b/2).
    $$
    In particular, we have
    \begin{equation*}
        \Pr_{\rndY \leftarrow F(a,b)}[\rndY \leq \theta] = 
        \frac{
            \int_0^{(a\theta)/(a\theta + b)}
            x^{a-1}
            (1-x)^{b-1}
            \,
            dx
        }{
            \int_0^1
            x^{a-1}
            (1-x)^{b-1}
            \,
            dx
        }.
    \end{equation*}
\end{lemma}

\subsection{Approximate $t$-designs}\label{ssec:t-design}

In our construction of unconditional $t$-copy OWSGs, we will use
approximate $t$-designs. Informally, an approximate $t$-design is a distribution over states such that $t$ copies of an output state is statistically close to $t$ copies of a Haar random state. One can think of $t$-designs as a quantum version of a $t$-wise independent distribution, or as a $t$-copy PRS with a statistical security guarantee.

\begin{definition}[Approximate $t$-Design, Definition 2.2 of~\cite{odonnell2023explicit}, rephrased]
    A probability distribution $S$ over $\mathbb{S}(2^n)$
    is an $\varepsilon$-approximate $t$-design if
    $$\left\lVert\Exp_{\ket{\psi} 
    \leftarrow
    \Haar(n)}\left[\ketbra{\psi}^{\otimes t}\right] -
    \Exp_{\ket{\psi} 
    \leftarrow
    S}\left[\ketbra{\psi}^{\otimes t}\right]\right\rVert_1 \leq
    \varepsilon,$$
    where $\lVert\cdot\rVert_1$ is the Schatten $1$-norm (or trace norm).
    We call $G$ an 
    \emph{efficient $\varepsilon$-approximate $t$-design}
    if $G$ is
    a 
    quantum algorithm running in time
    $\poly(n,m,t,\log(1/\eps))$
    which maps
    classical strings in $\{0,1\}^n$ to quantum
    states over $\mathbb{S}(2^m)$, and such that the output distribution of $G(\cdot)$ on a random $n$-bit string forms an
    $\varepsilon$-approximate $t$-design.
\end{definition}

Recently, it has been shown that approximate $t$-designs
with almost optimal seed exist unconditionally.

\begin{theorem}[Theorem 1.1 of ~\cite{odonnell2023explicit}, rephrased]
\label{thm:odonnell}
    For all $m,t,\varepsilon > 0$, there exists an efficient 
    $\varepsilon$-approximate $t$-design with input size $n = O(mt + \log(1/\varepsilon))$.
\end{theorem}

%% file: prs-owsg.tex
\section{One-way state generators from compressing pseudorandom states}
\label{sec:owsg-prs}

\input{short-owsg}

\subsection{Unconditional OWSGs from efficient approximate $t$-designs}

We first observe that any efficient $2^{-\lambda}$-approximate $t$-design is also a $t$-copy PRS. This follows definitionally, since $t$ copies from a $t$-design are statistically close to $t$ copies from a Haar random state.

\begin{proposition}
    Let $G$ be an efficient $2^{-\lambda}$-approximate $t$-design that maps $n$-bit strings to $m$-qubit pure states. Denote $\ket{\psi_k} := G(k)$. Then for any QPT adversary $\A$,
    $$\abs{\Pr_{k \leftarrow \{0, 1\}^n}\left[\A\left(\ket{\psi_k}^{\otimes t}\right) = 1\right] -
    \Pr_{\ket{\psi}\leftarrow\HaarN{m}}\left[\A\left(\ket{\psi}^{\otimes t}\right) = 1\right]} \leq
    \negl(\lambda).$$
\end{proposition}

In particular, the proposition holds definitionally even for \textit{statistical} adversaries, and so therefore must also hold for QPT adversaries. A simple corollary of Theorem~\ref{thm:prs-imply-strong-owsg} then shows that efficient approximate $t$-designs are also OWSGs.

\begin{corollary}\label{cor:t-design}
    Let $G$ be an efficient $\varepsilon$-approximate $t$-design mapping
    $n$ bits to $m = \omega(\log n)$ qubits. 
    Then $G$ is $(t, \varepsilon +
    \negl(n))$-one-way.
\end{corollary}

\Cref{thm:odonnell} with $\varepsilon = 2^{-\lambda}$ gives that, for any polynomial $t$, $t(\lambda)$-copy PRSs exist unconditionally, and thus \Cref{thm:prs-imply-strong-owsg} concludes that $t(\lambda)$-copy OWSGs also exist unconditionally.

Contrast this with the recent result of Khurana and Tomer \cite{dakshita23commitments}, where they show that $\Theta(n)$-copy OWSGs can be used to build quantum bit commitments (and thus require computational hardness \cite{LC97}). This raises the question of what is the largest number of copies $t$ (relative to the number of classical input bits $n$) for which OWSGs exist unconditionally. We show that the efficient approximate $t$-designs of \cite{odonnell2023explicit} approach this computational threshold up to a logarithmic factor.

\begin{corollary}
\label{cor:unconditional-owsg}
    For every function $\alpha = \alpha(n) = \omega(1)$,
    there exists a
    $\Theta\left(\frac{n}{\alpha \cdot \log n}\right)$-copy 
    OWSG.
\end{corollary}

\begin{proof}
    From Theorem~\ref{thm:odonnell}, we know that there exists some
    positive constant $c$ such that,
    for
    $n = c\cdot mt + c\log(1/\varepsilon)$,
    there is an efficient $\varepsilon$-approximate $t$-design
    mapping $n$ bits to $m$ qubits.

    Setting $\varepsilon = 2^{-\lambda}$, $n = 2c \cdot \lambda$,
    $t = \lambda/(\alpha \log n)$
    and $m = \alpha \cdot \log n$ 
    gives us 
    $$n = 2c \lambda = c \alpha t \log n + c\lambda
    = cmt + c\log(1/\eps),$$
    and thus
    we can build efficient $\varepsilon$-approximate $\Theta(n/(\alpha \log
    n))$-designs. 
    But since $m = \omega(\log n)$,
    using such a design also gives a $\Theta\left(\frac{n}{\alpha \log n}\right)$-copy strong OWSG from \Cref{cor:t-design}.
\end{proof}

%% file: short-owsg.tex
We construct both weak and strong $t$-copy one-way state generators from
$(t+1)$-copy compressing pseudorandom states. Before this the only known
reduction worked for expanding PRS mapping $n$-bit strings to $cn$-qubit
states~\cite{morimae2022one} for some $c > 1$. 
To generalize this reduction, we rely on the
following concentration inequality, which informally states that, with high
probability, any fixed state is unlikely to be close to Haar-random states.
We will then use this in Theorem~\ref{thm:prs-imply-owsg-gen} to bound how
well a OWSG inverter can distinguish outputs of a PRS from Haar-random states.

\begin{lemma}[Concentration of Haar States\protect\footnote{In
    \cite{Kretschmer21Quantum} the author refers to a closely related
inequality as following from ``standard concentration inequalities, or even
an explicit computation''. However we were unable to find an actual proof
of the inequality in either the citation listed \cite{brandao2016local}, or
in the paper which it cites \cite{hayden2006aspects}.
Consequently we decided to include a proof here for completeness.}]\label{lem:goodlem}
    Let $\ket{\phi_0}$ be any state of dimension $N = 2^n$. Then, for any $s > 0$,
    we have
    $$\Pr_{\ket{\psi}\sim \Haar(n)}\left[\abs{\braket{\phi_0|\psi}}^2 \geq \frac{1}{s}\right] \leq \left(\frac{s}{s+1}\right)^{N-1}.$$
\end{lemma}

\begin{proof}
    Since the Haar distribution is invariant under unitary transformations, without loss of generality we assume $\ket{\phi_0} = \ket{0}$. 

    Muller \cite{muller1959haar} showed that sampling from the Haar distribution is equivalent to sampling $\ket{\psi}$ as
    $$\ket{\psi} \propto \sum_{x} (\alpha_x+\beta_x i)\ket{x},$$
    where each $\alpha_x,\beta_x$ is sampled according to the standard Gaussian with expectation 0 and standard deviation 1. Then, we define the random variable $Y$ as 
    $$Y := \frac{\alpha_0^2 + \beta_0^2}{\sum_{x \neq 0} \alpha_x^2 + \beta_x^2}.$$
    Expanding out the inner product gives us
    $$\abs{\braket{0|\psi}}^2 = \frac{\alpha_0^2 + \beta_0^2}{\sum_x (\alpha_x^2 + \beta_x^2)} \leq Y.$$
    But observe that each $\alpha_x, \beta_x$ is sampled independently, and
    so $Y$ is distributed as the ratio of two chi-squared random variables.
    So, we see that $Y$ is sampled as a (scaled) $F$-distribution as
    follows:
    $$Y \sim \frac{\chi^2(2)}{\chi^2(2N-2)} \sim \frac{2}{2N-2} \cdot F(2,2N-2).$$
    We conclude by \Cref{lem:cumulative-function-fs} that
    \begin{align*}
        \begin{split}
            \Pr\left[Y \geq \frac{1}{s}\right] 
            &= \Pr\left[F(2, 2N-2) \geq \frac{N-1}{s}\right] \\
            &= 1 - \frac{\int_0^{\frac{1}{s+1}} (1-x)^{N-2}dx}{\int_0^1 (1-x)^{N-2}dx}\\
            &= 1 - \frac{\left[-\frac{1}{N-1}\cdot (1-x)^{N-1}\right]_0^{1/(s+1)}}{\left[-\frac{1}{N-1}\cdot (1-x)^{N-1}\right]_0^{1}} \\
            &= 1 - \frac{\left[(1-x)^{N-1}\right]_0^{1/(s+1)}}{\left[(1-x)^{N-1}\right]_0^{1}} \\
            &= 1 + \left(\left(1-\frac{1}{s + 1}\right)^{N-1} - 1\right) \\
            &= \left(\frac{s}{s+1}\right)^{N-1}, \\
        \end{split}
    \end{align*}
    where
    $[f(x)]_0^1 = f(1) - f(0)$. 
    Since
    $$\Pr_{\ket{\psi}\sim \text{Haar}(N)}\left[\abs{\braket{\phi_0|\psi}}^2 \geq \frac{1}{s}\right] \leq \Pr\left[Y \geq \frac{1}{s}\right],$$
    we are done.
\end{proof}

Using this lemma we first show a general result 
stating that state generators that are pseudorandom must also be one-way. 
We then apply the result
in two
different parameter regimes.
While \cite{Kretschmer21Quantum} only claims that 
PRSs
with $m = n$ output bits
can be broken
by $\PP$, \cite{ananth2022pseudorandom} calls out that the proof can be
extended to the case when $m \geq \log n + c$. 
For this regime, we show that
PRSs are weak OWSGs. We then show that slightly less compressing PRSs ($m = \omega(\log n)$) are strong OWSGs. 
In all these three results, the number of copies
falls by 1 moving from a PRS to a OWSG. While we end up focusing on
sublinear-copy PRSs in the next subsection, these reductions work in the
default many-copy setting considered in most papers on quantum
cryptographic primitives.

\begin{lemma}

\label{thm:prs-imply-owsg-gen}
    For all $f(n)$ and for
    $$\delta = 2^n \cdot \left(\frac{f(n)}{f(n)+1}\right)^{(2^m - 1)} +
    \frac{1}{f(n)},$$
    if 
    $G : k \mapsto \ket{\phi_k}$
    is a state generator taking $n$-bit
    strings to $m$-qubit pure states which is $\left(t +
    1,\eps\right)$-pseudorandom, then it is also $(t,\eps +
    \delta)$-one way.
\end{lemma}

\begin{proof}
    For the sake of contradiction, assume that there exists an adversary
    $\A$ that can succeed with probability larger than $\eps + \delta$ in the $t$-copy
    one-wayness game (\Cref{def:owsg}). We will construct a new adversary $\A'$ for the
    pseudorandomness game as follows:\\

    \begin{algorithm}
        \caption{Adversary $\mathcal{A}'$ in the pseudorandomness game}
        \label{alg:cap}
        \hspace*{\algorithmicindent} \textbf{Input}: $\ket{\psi}^{\otimes (t+1)}$. \\
        \hspace*{\algorithmicindent} \textbf{Output}: $0$ if $\ket{\psi}$ is pseudorandom, $1$ if $\ket{\psi}$ is Haar random. \\
        \begin{algorithmic}[1]
            \STATE Run $\A$ on the first $t$ copies and obtain $k' \leftarrow \A(\ket{\psi}^{\otimes t})$ \\
            \STATE
            Measure the last copy of $\ket{\psi}$ 
            in the basis 
            $\left\{\ketbra{\phi_{k'}}, I - \ketbra{\phi_{k'}}\right\}$ \\
            \RETURN 1 if the result is $\ketbra{\phi_{k'}}$, else
            \textbf{return}
            0.
        \end{algorithmic}
    \end{algorithm}

    Observe that when $\A'$ is given a pseudorandom input, it outputs $1$ with probability at least
    $$\Pr_{k \leftarrow \K}\left[\A'\left(\ket{\phi_k}^{\otimes (t+1)}\right)
    = 1\right] = \E_{k\leftarrow \K}
    \left[\abs{\braket{\phi_k|\phi_{k'}}}^2 \middle| k' \leftarrow
    \A\left(\ket{\phi_k}^{\otimes t}\right)\right] > \eps + \delta,$$
    where the key space $\calk$ is equal to $\blt^n$ here.
    Thus, it suffices remains to bound the probability that $\A'$ detects a Haar random state:
    $$\Pr_{\ket{\psi}\leftarrow\Haar(m)} \left[\A'\left(\ket{\psi}^{\otimes
    (t+1)}\right) = 1\right] \leq \delta.$$
    By construction, we have
    \begin{align*}
        \Pr_{\ket{\psi}\leftarrow\Haar(m)} \left[\A'\left(\ket{\psi}^{\otimes (t+1)}\right) = 1\right] 
        =&~\E_{\ket{\psi}\leftarrow \Haar(m)} \left[\abs{\braket{\psi|\phi_{k}}}^2 \middle| k \leftarrow \A\left(\ket{\psi}^{\otimes t}\right)\right]\\
        =&~\int_{\mathbb{S}(2^m)} d\mu(\psi) \cdot \sum_{k \in \{0, 1\}^n} \Pr\left[k \leftarrow \A\left(\ket{\psi}^{\otimes t}\right)\right]\cdot  \abs{\braket{\psi|\phi_{k}}}^2\\
        \leq&~\int_{\mathbb{S}(2^m)} d\mu(\psi) \cdot \max_{k \in \{0, 1\}^n} \abs{\braket{\psi|\phi_k}}^2.
    \end{align*}
    Where $\mu(\psi)$ is the Haar measure on the space of $m$-qubit pure
    states. We will now partition the set of $m$-qubit pure states into the
    states that are `close' to a pseudorandom state $\ket{\phi_k}$, and the
    states that are `far' from all pseudorandom states. Formally, we define
    $$A_f := \left\{\ket{\psi} \in \mathbb{S}(2^m) \middle| \max_k \abs{\braket{\psi|\phi_k}}^2 \geq \frac{1}{f(n)}\right\}.$$
    The set of states that are `far' from all pseudorandom states is its complement,
    $$B_f := \left\{\ket{\psi} \in \mathbb{S}(2^m) \middle| \max_k \abs{\braket{\psi|\phi_k}}^2 < \frac{1}{f(n)}\right\}.$$
    We proceed with computing the integral separately for the two sets:

    \begin{align*}
        \int_{\mathbb{S}(2^m)} d\mu(\psi) \cdot \max_{k \in \{0, 1\}^n} \abs{\braket{\psi|\phi_k}}^2
        =&~\int_{A_f} d\mu(\psi)\cdot \max_{k \in \{0, 1\}^n} \abs{\braket{\psi|\phi_k}}^2 
        \\&+ \int_{B_f} d\mu(\psi)\cdot \max_{k \in \{0, 1\}^n} \abs{\braket{\psi|\phi_k}}^2 \\
        \leq&~\int_{A_f} d\mu(\psi) + \max_{\substack{\ket{\psi} \in B_f \\ k \in \{0, 1\}^n}} \abs{\braket{\psi|\phi_k}}^2\\
        <&~\Pr_{\ket{\psi} \leftarrow \Haar(m)}\left[ \exists k \middle| \abs{\braket{\psi|\phi_k}}^2 \geq \frac{1}{f(n)}\right] + \frac{1}{f(n)}\\
        \leq&~ \sum_{k \in \K}\Pr_{\ket{\psi}\leftarrow \Haar(m)} \left[ \abs{\braket{\psi|\phi_k}}^2 \geq \frac{1}{f(n)}\right] + \frac{1}{f(n)}.
    \end{align*}
    \Cref{lem:goodlem} implies that
    $$\Pr_{\ket{\psi}\leftarrow\Haar(m)} \left[\A'\left(\ket{\psi}^{\otimes
    (t+1)}\right) = 1\right] \leq 2^n\cdot \left(\frac{f(n)}{1+f(n)}\right)^{2^m - 1} + \frac{1}{f(n)} = \delta.$$
    We conclude
    \begin{equation*}
        \abs{\Pr_{k \leftarrow \K} \left[\A'\left(\ket{\psi}^{\otimes
        (t+1)}\right) = 1\right] - \Pr_{\ket{\psi} \leftarrow\Haar(m)}
    \left[ \A'\left(\ket{\psi}^{\otimes (t+1)}\right) = 1\right]} > \eps.
        \qedhere
    \end{equation*}
\end{proof}

Using this general result we can specify to the two theorems below.

\begin{theorem}[PRSs imply OWSGs]
\label{thm:prs-imply-weak-owsg}
     If 
     $G$
     is a state generator taking $n$-bit strings to $m > \log n + c$ qubit states (with $c \geq 1$)
     which is $(t+1,\eps)$-pseudorandom, then $G$ is also $(t,\eps + 3/4)$-one-way.
\end{theorem}
\begin{proof}
    Take \Cref{thm:prs-imply-owsg-gen} with $f(n) = 2$, and $m > \log n + c$. We get that
    \begin{align*}
        \delta &\leq 2^n \cdot \left(\frac{2}{3}\right)^{n\cdot 2^c - 1} + \frac{1}{2} \\
        &= \frac{3}{2} \cdot \left(\frac{2^{2^c+1}}{3^{2^c}}\right)^{n} + \frac{1}{2}.
    \end{align*}

    Since $c > 1$, the left term becomes less than $1/4$ for sufficiently large $n$, and thus $\delta  < 3/4$.
\end{proof}

Using amplification arguments from \cite{morimae2022quantum},
these $(t, \eps + 3/4)$-OWSGs can be used to construct strong OWSGs with negligible probability of inversion.

This reduction is nearly optimal. \cite{brakerski2022computational} showed that there exist statistical many-copy PRSs with $c\log n$ output bits for some $c < 1$. As shown later strong OWSGs must imply computational assumptions (e.g. $\BQP \neq \PP$) so they cannot be implied by statistical or information theoretic primitives. Consequently there can be no proof that PRSs with output shorter than $c \log n$ qubits imply OWSGs, meaning our proof is optimal up to multiplicative constant factors.

We can also show that PRSs with 
superlogarithmic output
are themselves strong OWSGs.

\begin{theorem}[PRSs are strong OWSGs]
\label{thm:prs-imply-strong-owsg}
     If 
     $G$
     is a state generator taking $n$-bit
     strings to 
     $m = \omega(\log n)$-qubit
     states which is
     $(t+1,\eps)$-pseudorandom, then it is also 
     $(t,\eps + \negl(n))$-one way.
\end{theorem}

\begin{proof}
    Take \Cref{thm:prs-imply-owsg-gen} and set 
    $f(n) = \frac{2^m-1}{2^{m/2}} - 1$. 
    Note that $f(n) = 2^{\Omega(m)} = n^{\omega(1)}$.
    We get 
    \begin{equation*}
    \begin{split}
        \delta & = 2^n \cdot \left(\frac{f(n)}{f(n)+1}\right)^{2^m - 1}
        + \frac{1}{f(n)}\\
        & \leq 2^n \cdot \left(1 - \frac{1}{f(n)+1}\right)^ {2^m - 1} + \negl(n)\\
        & \leq 2^n \cdot\left(\frac{1}{e}\right)^{2^{m/2}} + \negl(n)\\
        & = 2^{n-\Omega(n^{\omega(1)})} + \negl(n)\\
        & = \negl(n).
        \qedhere
        \end{split}
    \end{equation*}
\end{proof}

%% file: break-owsg-pp-oracle.tex
\section{Breaking one-way state generators with a PP oracle}
\label{sec:break-pp}

In this section, we show how to break one-way state generators (OWSGs)
with a $\PP$ oracle. 
Since one-way state generators imply one-way puzzles
(\Cref{thm:owsg-to-owpuz}, due to~\cite{dakshita23commitments}),
it suffices to show how to break one-way puzzles using a $\PP$
oracle.
Recall that a one-way puzzle is a pair of sampling and verification
quantum algorithms
$(\Samp, \Ver)$~(see \Cref{def:owpuz}).
The algorithm $\Samp(1^n)$ outputs a pair $(k,s)$, where $k$ is referred to
as the key and $s$ as the puzzle.
To break the one-way puzzle, it suffices to create a quantum algorithm
$\cala$
that, given a puzzle $s$, returns a key $k'$ such that
$\Ver(k', s)$ is accepted with non-negligible probability.

Our strategy to construct $\cala$ is as follows.
Given a puzzle $s$ sampled according to 
$\Samp$,
we will sample a key $k$ 
according to the conditional distribution of keys that are sampled together
with $s$.
This will suffice to break the one-way puzzle, as $\Ver$
accepts pairs from $\Samp$ with $(1-\negl)$-probability.
To sample from this distribution, 
we first show how to use the $\PP$ oracle to estimate the conditional probability,
and finally how to sample according to a distribution that is close to the true conditional distribution.

We view our strategy as inspired by Kretschmer's \cite{Kretschmer21Quantum} idea for breaking pseudorandom state generators with a $\PP$ oracle, together with a search-to-decision reduction for $\PP$.

\begin{lemma}
    \label{lem:pp-oracle}
    Let $S$ be a (uniform) quantum polynomial time algorithm that outputs $n$ bits, denoted as a pair of a string $x$, and a bit $b$. There exists a poly-time quantum algorithm $\mathcal{A}$ and a $\PP$ language $\mathcal{L}$ such that $\mathcal{A}^\mathcal{L}$ can estimate the distribution of bit $b$, conditioned on the output string $x$. Formally,
    
    $$p_{x, b} - \frac{1}{n^2} \leq \mathcal{A}^\mathcal{L}(S, 1^n, x, b) \leq p_{x, b} + \frac{1}{n^2}$$
    where $p_{x, b} = \Pr\left[S(1^n) = (x, b) | S(1^n) \in \{(x, 0), (x, 1)\}\right]$.
\end{lemma}
\begin{proof}
    Our quantum algorithm $\mathcal{A}$ will take as input the algorithm
    $S$, a unary description of the length $1^n$, the outputs $(x, b)$, and will estimate the conditional
    probability of $b$, conditioned on the first output of $S$ being $x$.
    
    \paragraph{Definition of the $\PP$ language.}
    We will describe the $\PP$ language $\mathcal{L}$ in terms of a
    $\textsf{PromisePostBQP}$ algorithm $\mathcal{B}(S, x, t)$.
    This will define a promise problem which is computed by
    $\mathcal{B}(S,x,t)$.
    By the equivalence
    $\PromPostBQP = \PromPP$, this gives us a promise problem in $\PromPP$.
    Since $\PromPP$ is a syntactic class, this can be extended to a
    language $\call \in \PP$.
    However, later on our algorithm $\cala$ will only depend on the
    behaviour of the oracle to $\call$ on inputs that satisfy the promise
    condition of the promise problem defined by $\mathcal{B}(S,x,t)$.
    
    We first assume without loss of generality that any measurements in $S$ are delayed until the very end of the algorithm.
    The algorithm $\mathcal{B}$ will simulate $S$ and postselect
    on the output measurements matching $x$. Conditioned on
    postselection succeeding, the output register for $b$ will contain the pure state  $\sqrt{p_{x, 0}}\ket{0} + \sqrt{p_{x, 1}}\ket{1}$. 
    Then $\mathcal{B}$ measures the register of $b$,
    and repeats this procedure  $r = \Theta(n^4)$ times,
    outputting $1$ if the 
    measurements
    yields
    $1$ at least a $\frac{t}{2n^2}$ fraction of the time.
    This ends the description of the $\PromPostBQP$ algorithm
    $\calb(S,x,t)$.

    Which inputs $(S,x,t)$ are accepted by $\calb$?
    Let 
    $\rndb_1, \dots, \rndb_r$ 
    denote the random variables that correspond to the value of the $b$ register at each simulation of $S$ by $\mathcal{B}$. 
    We define $\rndB = \sum_i \rndb_i$ to be their sum. 
    Note that $\rndB/r$ 
    is the fraction of measurements that yield a $1$
    in the algorithm $\calb$, 
    and recall that we accept if this fraction is at least
    $t/(2n^2)$.
    By 
    standard concentration inequalities (e.g., Chebyshev's), 
    when the number of repetitions $r$ is at least $\Theta(n^4)$, then
    $$\Pr_{\rndb_1,\dots,\rndb_r}\left[\left|\frac{\rndB}{r} - p_{x, 1}\right| > \frac{1}{4n^2}\right] < \frac{1}{3}.$$
    This implies that, 
    if $p_{x,1} \geq \frac{t}{2n^2} + \frac{1}{4n^2}$,
    then
    $\frac{\rndB}{r} > \frac{t}{2n^2}$
    with probability $2/3$.
    Moreover, if
    $p_{x,1} \leq \frac{t}{2n^2} - \frac{1}{4n^2}$,
    then
    $\frac{\rndB}{r} < \frac{t}{2n^2}$
    with probability $2/3$.
    We conclude
    \begin{equation}
        \label{eq:prob-est}
        \begin{cases}
            (S,x,t) \in \call,
            &\text{if }
            \quad
            p_{x,1} \geq \frac{t}{2n^2} + \frac{1}{4n^2},
            \\
            (S,x,t) \not\in \call,
            &\text{if }
            \quad
            p_{x,1} \leq \frac{t}{2n^2} - \frac{1}{4n^2}.
        \end{cases}
    \end{equation}
    Note that this only gives us information about the output of $\calb$
    on inputs $(S,x,t)$ such that
    $\abs{p_{x,1} - \frac{t}{2n^2}} \geq \frac{1}{4n^2}$.
    Henceforth we will 
    call this inequality the \emph{promise condition}.
    We remark that, if
    $(S,x,t)$ does not satisfy the promise condition,
    then
    $(S,x,t+1)$ satisfies it.
    Indeed,
    if
    $\abs{p_{x,1} - \frac{t}{2n^2}} < \frac{1}{4n^2}$,
    we have
    \begin{equation*}
        p_{x,1} 
        \leq
        \frac{t}{2n^2} + \frac{1}{4n^2}
        =
        \frac{t+1}{2n^2} - \frac{1}{2n^2} + \frac{1}{4n^2}
        =
        \frac{t+1}{2n^2} - \frac{1}{4n^2},
    \end{equation*}
    and therefore
    $\abs{p_{x,1} - \frac{t+1}{2n^2}} \geq \frac{1}{4n^2}$.
    In particular, we obtain that
    $(S,x,t+1) \notin \call$.

    \paragraph{Description of $\cala^\call$.}
    Given access to this $\PP$ language, $\mathcal{A}$ will query the oracle on the inputs 
    $\{(S, x, 0), \dots, (S, x, 2n^2)\}$.
    Then $\mathcal{A}$ will output $\frac{t}{2n^2}$,
    for the smallest $t$ such that 
    $(S,x,t)$ is rejected by the oracle 
    (if $(S, x, t)$
    is accepted
    for all $t$,
    then $\mathcal{A}$ outputs
    $1$). 
    Note that, if 
    $(S,x,t-1)$ is accepted by the oracle,
    and $(S,x,t)$ is rejected, then one of the following three things must
    have happened:
    \begin{enumerate}
        \item 
            $(S,x,t-1) \in \call$ and
            $(S,x,t) \not\in \call$, or
        \item
            $(S,x,t-1)$ does not satisfy the promise condition,
            and
            $(S,x,t) \not\in \call$, or
        \item
            $(S,x,t-1) \in \call$ and
            $(S,x,t)$ does not satisfy the promise condition.
    \end{enumerate}
    As we observed above, it never occurs that
    $(S,x,t-1)$ does not satisfy the promise condition,
    and
    $(S,x,t) \in \call$.
    By inspection, using (\ref{eq:prob-est}), in all of those three cases we get
    that
    the value $\frac{t}{2n^2}$ is 
    an additive
    approximation to $p_{x, 1}$ with error at most $\frac{1}{n^2}$.
    Additionally, since $p_{x, 0} = 1 - p_{x, 1}$, the value $\frac{2n^2 -
    t}{2n^2}$ is an additive approximation for $p_{x, 0}$.
\end{proof}

When $\Samp(1^n) \to (k,s)$ is 
is a uniform quantum polynomial-time algorithm
outputting a pair of classical strings
$(k,s)$,
we will denote by
$\Dkeycond{s'}$
the distribution of keys $k$ output by
$\Samp(1^n)$, conditioned on the puzzle being $s'$. 

\begin{lemma}
    \label{lem:sample-bit-by-bit}
    Let $\Samp$ be a (uniform) quantum polynomial time algorithm 
    such that
    $\Samp(1^n)$ outputs a
    pair of classical strings $(k, s)$, denoted as the key and the puzzle
    respectively, where $k \in \blt^n$.
    There exists a poly-time
    quantum algorithm $\mathcal{A}$ and a $\PP$ language $\mathcal{L}$ such
    that $\mathcal{A}^\mathcal{L}$ takes as input a puzzle $s'$ and outputs
    a key $k'$, 
    and whose distribution 
    has total variation distance at most $1/n$ from $\Dkeycond{s'}$.
    In other words,
    we have
    $$\left(k' \mid k' \leftarrow \mathcal{A}^\mathcal{L}(1^n, s')\right) \approx_{\frac{1}{n}} 
    \Dkeycond{s'}.$$
\end{lemma}

\begin{proof}
    On a high level, the algorithm $\mathcal{A}$ will output a key $k'$ by
    sampling its bits one by one. At the $i^{th}$ iteration, it will use
    the first $i-1$ bits of $k'$ to estimate the distribution of the
    $i^{th}$ bit using a $\PP$ oracle, as in \Cref{lem:pp-oracle}. Then it
    will sample the $i^{th}$ bit of $k'$ according to this estimated
    distribution. The $\PP$ language $\mathcal{L}$ that we use is the same
    as the one shown to exist in that lemma.

    In more detail, let us define 
    the sequence 
    $\{ \Sampx{i} \}_{i \in [n]}$
    of algorithms based on $\Samp$, where $\Sampx{i}$ is the (uniform) quantum polynomial
    time algorithm that simulates $\Samp$, but only outputs the puzzle $s$ and
    the first $i$ bits of the key $k$.

    On input $(1^n, s')$, the algorithm $\mathcal{A}$ will proceed in $n$ iterations. 
    In the first iteration, it will use \Cref{lem:pp-oracle} on 
    $\Sampx{1}$ with output $s'$ and estimate (up to $\frac{1}{n^2}$ additive
    error) the probability that the first bit of the key is $1$,
    conditioned on the puzzle $s'$. Call this estimate $\tilde{p}_1$.
    The algorithm
    will then sample the first bit of $k'$ according to the
    Bernoulli distribution defined by $\tilde{p}_1$.

    Now $\mathcal{A}$ proceeds by sampling the remaining bits. In the
    $i^{\rm th}$ iteration, it uses \Cref{lem:pp-oracle} with the sampler
    $\Sampx{i}$ and outputs an estimate $\tilde{p}_i$ of the probability
    that the $i^{\rm th}$ bit of $k'$ is $1$, conditioned on the puzzle $s'$
    and the first $i-1$ bits of $k'$. Then it samples the $i^{\rm th}$ bit
    according to the estimated distribution. 
    Finally, $\mathcal{A}$ outputs $k'$
    after the end of the $n^{\rm th}$ iteration.

    It remains to show that the output distribution of $\mathcal{A}$ is
    close to $\Dkeycond{s'}$, the output distribution of $\Samp(1^n)$ conditioned on the
    puzzle $s'$. We will show this via a hybrid argument.

    Let $\cald_0 := \Dkeycond{s'}$ be the true distribution of the key $k$ conditioned on
    the puzzle $s'$. 
    We define $n$ hybrid distributions $\cald_i$ for $i \in \{1, \dots,
    n\}$ on keys. 
    The hybrid $\cald_i$ runs in $n$ iterations and in each
    iteration samples the next bit of the key. In the first $i$ iterations
    (that correspond to the first $i$ bits),
    the distribution $\cald_i$ uses the estimated
    probabilities $\tilde{p_1},\dots,\tld{p_i}$. The final $n-i$ bits are sampled according
    to the true conditional probabilities (i.e., according to
    $\Dkeycond{s'}$, conditioned on the outcome of the previous bits).
    Note that
    $\cald_n$ now corresponds to the output distribution of our algorithm
    $\mathcal{A}$. We show that every two consecutive distributions are at
    most $\frac{1}{n^2}$ far in total variation distance,
    and thus the total distance between the
    true distribution of the key and the output distribution
    of $\mathcal{A}$ is
    at most $\frac{1}{n}$ from the triangle inequality.

    \begin{claim}
        For every $i \in \set{0,\dots,n-1}$, we have
        $$\dTV(\cald_i, \cald_{i+1}) \leq \frac{1}{n^2}.$$
    \end{claim}

    \begin{proof}
        Before diving into the proof, we introduce some useful notation.
        We use $k_{[x, y]}$ to denote the length-$(y-x+1)$ substring that includes bits $\{x, \dots, y\}$ of $k$.
        Additionally, even though $\cald_i$ is a distribution over
        $n$-bit strings, we will abuse notation and consider the probability assigned to substrings of the form $k_{[x, y]}$. In that case, we write
        $$\cald_i(k_{[x, y]}) := \Pr_{k' \leftarrow \cald_i}\left[k'_{[x, y]} = k_{[x, y]}\right].$$
        We will also consider the probability assigned to substrings $k[x, y]$, conditioned on a prefix $k_{[1, x]}$, in which case we write
        $$\cald_i\left(k_{[x, y]} \middle| k_{[1, x-1]}\right) := \Pr_{k' \leftarrow \cald_i}\left[k'_{[x, y]} = k_{[x, y]} \middle| k'_{[1, x-1]} = k_{[1, x-1]}\right].$$
    
        With this notation in place, a direct calculation suffices
        for the $i = 0$ case. In particular, observe that
        \begin{align*}
            \cald_1(k) &= \cald_1\left(k_{[1, 1]}\right)\cdot \cald_1\left( k_{[2, n]}\middle|k_{[1, 1]}\right) \\
            &= \cald_1\left(k_{[1, 1]}\right)\cdot \cald_0\left( k_{[2,
            n]}\middle|k_{[1, 1]}\right),
        \end{align*}
        and similarly $\cald_0(k) = \cald_0\left(k_{[1, 1]}\right)\cdot \cald_0\left( k_{[2, n]}\middle|k_{[1, 1]}\right)$. 
        Note that, because of
        \Cref{lem:pp-oracle},
        we have
        $\left|\cald_0\left(k_{[1, 1]}\right) - \cald_1\left(k_{[1,
        1]}\right)\right| \leq 1/n^2$.
        The total variation distance satisfies:
        \begin{align*}
            \dTV(\cald_0, \cald_1) &= \frac{1}{2}\sum_{k \in \{0, 1\}^n} \left|\cald_0(k) - \cald_1(k)\right| \\
            &= \frac{1}{2}\sum_{k \in \{0, 1\}^n} \left|\cald_0\left(k_{[1, 1]}\right)\cdot \cald_0\left( k_{[2, n]}\middle|k_{[1, 1]}\right) - \cald_1\left(k_{[1, 1]}\right)\cdot \cald_0\left( k_{[2, n]}\middle|k_{[1, 1]}\right)\right| \\
            &\leq \frac{1}{2}
            \sum_{k \in \{0, 1\}^n} 
            \left|\cald_0\left(k_{[1, 1]}\right) - \cald_1\left(k_{[1, 1]}\right)\right| \cdot \left|\cald_0\left( k_{[2, n]}\middle|k_{[1, 1]}\right)\right| \\
            &\leq \frac{1}{2n^2}\sum_{k \in \{0, 1\}^n} 
            \left|\cald_0\left( k_{[2, n]}\middle|k_{[1, 1]}\right)\right| \\
            &\leq \frac{1}{n^2}.
        \end{align*}

        Let us now consider the distributions
        $\cald_{i}, \cald_{i+1}$. 
        They both sample bits $1$ up to $i$ using the estimated probabilities,
        and bits $i+2$ up to $n$ with the true conditional probabilities. 
        Write $\cald_{i+1}$ as
        \begin{align*}
            \cald_{i+1}(k) &= \cald_{i+1}\left(k_{[1, i]}\right)
                \cdot  \cald_{i+1}\left(k_{[i+1, i+1]} \middle| k_{[1, i]}\right)
                \cdot \cald_{i+1}\left(k_{[i+2, n]} \middle| k_{[1, i+1]}\right) \\
            &= \cald_{i+1}\left(k_{[1, i]}\right)
                \cdot  \cald_{i+1}\left(k_{[i+1, i+1]} \middle| k_{[1, i]}\right)
                \cdot \cald_{i}\left(k_{[i+2, n]} \middle| k_{[1, i+1]}\right)
        \end{align*}
        and similarly $\cald_i$ as
        \begin{align*}
            \cald_{i}(k) &= \cald_{i}\left(k_{[1, i]}\right)
                \cdot  \cald_{i}\left(k_{[i+1, i+1]} \middle| k_{[1, i]}\right)
                \cdot \cald_{i}\left(k_{[i+2, n]} \middle| k_{[1, i+1]}\right) \\
            &= \cald_{i+1}\left(k_{[1, i]}\right)
                \cdot  \cald_{i}\left(k_{[i+1, i+1]} \middle| k_{[1, i]}\right)
                \cdot \cald_{i}\left(k_{[i+2, n]} \middle| k_{[1, i+1]}\right).
        \end{align*}
        Note that, because of
        \Cref{lem:pp-oracle},
        we have
        \begin{equation*}
            \left|\cald_i\left(k_{[i+1, i+1]}\middle|k_{[1, i]}\right) -
            \cald_{i+1}\left(k_{[i+1, i+1]}\middle|k_{[1, i]}\right)\right|
            < 1/n^2.
        \end{equation*}
        The total variation distance satisfies:
        \begin{align*}
            &\dTV(\cald_i, \cald_{i+1}) \\ 
            &= \frac{1}{2}\sum_{k \in \{0, 1\}^n} \left|\cald_i(k) - \cald_{i+1}(k)\right| \\
            &\leq \frac{1}{2}\sum_{k \in \{0, 1\}^n}
            \cald_{i+1}\left(k_{[1, i]}\right)
            \cdot
            \left|\cald_i\left(k_{[i+1, i+1]}\middle|k_{[1, i]}\right) - \cald_{i+1}\left(k_{[i+1, i+1]}\middle|k_{[1, i]}\right)\right|
            \cdot \cald_{i}\left(k_{[i+2, n]} \middle| k_{[1, i+1]}\right) \\
            &\leq \frac{1}{2n^2}\sum_{k \in \{0, 1\}^n}
            \cald_{i+1}\left(k_{[1, i]}\right)
            \cdot \cald_{i}\left(k_{[i+2, n]} \middle| k_{[1, i+1]}\right) \\
            &= \frac{1}{2n^2}\sum_{k_{[1, i]} \in \{0, 1\}^i}
            \cald_{i+1}\left(k_{[1, i]}\right)
            \sum_{k_{[i+1, i+1]} \in \{0, 1\}} \sum_{k_{[i+2, n]} \in \{0, 1\}^{n-i-1}}
            \cald_{i}\left(k_{[i+2, n]} \middle| k_{[1, i+1]}\right) \\
            &= \frac{1}{n^2}\sum_{k_{[1, i]} \in \{0, 1\}^i}
            \cald_{i+1}\left(k_{[1, i]}\right) \\
            &= \frac{1}{n^2}.
            \qedhere
        \end{align*}
    \end{proof}
    As observed above,
    by the triangle inequality we obtain from
    the Claim
    that $\cald_n$, which is the
    output distribution of our algorithm $\cala$, has total variation
    distance at most $1/n$ from 
    $\Dkeycond{s'}$.
\end{proof}

\begin{theorem}
    \label{thm:pp-break-owpuzz}
    For any one-way puzzle
    $(\Ver, \Samp)$, there exists a $\PP$ language $\mathcal{L}$, and a poly-time quantum algorithm $\mathcal{A}^\mathcal{L}$ such that
    $$\Pr_{(k,s) \leftarrow \Samp(1^n)}
        \left[\Ver\left(\mathcal{A}^\mathcal{L}(s),s\right) = \top\right]
        \geq \frac{1}{2}.$$
\end{theorem}

\begin{proof}
    From the correctness property of the OWPuzzle, it holds that
    $$\Pr_{(k, s) \leftarrow \Samp(1^n)}\left[\Ver(k, s) = \top\right] \geq
    1 - \negl(n).$$
    Recall that
    $\Dkeycond{s'}$
    is the distribution over keys
    output by $\Samp$, conditioned on the puzzle being equal to $s'$.
    It is clear that sampling the puzzle first, and then the key
    does not change their joint distribution, and thus
    $$\Pr_{\substack{(k, s) \leftarrow \Samp(1^n) \\ k' \leftarrow \Dkeycond{s}}}
    \left[\Ver(k', s) = \top\right] \geq 1 - \negl(n).$$
    Given a puzzle $s$, \Cref{lem:sample-bit-by-bit} implies that there
    exists a quantum polynomial-time algorithm $\mathcal{A}$ and a $\PP$
    language $\mathcal{L}$, such that $\mathcal{A}^\mathcal{L}(s')$
    outputs a key $k'$ according to a distribution $\tilde{D}_{s'}$
    such that
    \begin{equation}
        \label{eq:dtv}
        \tilde{D}_{s'} \approx_{1/n} \Dkeycond{s'}.
    \end{equation}
    We have
    \begin{align*}
        \Pr_{\substack{(k, s) \leftarrow \Samp(1^n) \\
            k'
        \leftarrow \tilde{D}_s}}\left[\Ver(k', s) =
        \top\right]
        &=
        \sum_{s'}
        \Pr_{\substack{(k, s) \leftarrow \Samp(1^n) \\
            k'
        \leftarrow \tilde{D}_s}}
        \left[\Ver(k', s) = \top\right | s = s']
        \cdot
        \Pr_{(k,s) \leftarrow \Samp(1^n)}[s = s'].
    \end{align*}
    Now note that, conditioned on $s=s'$,
    the event
    $\set{\Ver(k', s) = \top}$
    depends only on $k' \leftarrow \tld{D}_{s'}$,
    and thus we may use (\ref{eq:dtv}) to get
    \begin{align*}
        \Pr&_{\substack{(k, s) \leftarrow \Samp(1^n) 
                \\
            k'
        \leftarrow \tilde{D}_s}}
        \left[\Ver(k', s) =
        \top\right]
        \\&\geq
        \sum_{s'}
        \left(
            \Pr_{\substack{(k, s) \leftarrow \Samp(1^n) \\
            k'
        \leftarrow \Dkeycond{s}}
        }
        \left[\Ver(k', s) = \top\right | s = s'] - 1/n
        \right)
        \cdot
        \Pr_{(k,s) \leftarrow \Samp(1^n)}[s = s']
        \\&
        =
        \sum_{s'}
        \left(
            \Pr_{\substack{(k, s) \leftarrow \Samp(1^n) \\
            k'
        \leftarrow \Dkeycond{s}}
        }
        \left[\Ver(k', s) = \top\right | s = s'] 
        \Pr_{(k,s) \leftarrow \Samp(1^n)}[s = s']
        \right)
        -
        (1/n)
        \cdot
        \Pr_{(k,s) \leftarrow \Samp(1^n)}[s = s']
        \\&
        =
            \Pr_{\substack{(k, s) \leftarrow \Samp(1^n) \\
            k'
        \leftarrow \Dkeycond{s}}
        }
        \left[\Ver(k', s) = \top\right] 
        -
        1/n
        \geq
        1-1/n-\negl(n)
        > 1 - 2/n.
    \end{align*}
    This completes our argument.
\end{proof}

Now the desired result follows by combining \Cref{thm:pp-break-owpuzz} with
\Cref{thm:owsg-to-owpuz} (Theorem 4.2 of \cite{dakshita23commitments}).

\begin{corollary}
    \label{cor:pp-break-owsg}
    For any OWSG $G$ of
    $n$-qubit states with key space $\calk$
    and with security parameter $\lambda$, there exists a $\PP$
    language $\mathcal{L}$, a $\textsf{poly}(\lambda)$-time quantum
    algorithm $\mathcal{A}^\mathcal{L}$, and $t = \textsf{poly}(\lambda)$
    such that $$\E_
    {   \substack{k \leftarrow \calk \\ 
        k' \leftarrow \mathcal{A}^{\mathcal{L}}\left(\ket{\phi_k}^{\otimes t}\right)}
    }
    \left[
    \abs{\braket{\phi_k |\phi_{k'}}}^2
    \right] \geq \frac{1}{2}.$$
\end{corollary}
\begin{proof}
    The proof of Theorem 4.2 of~\cite{dakshita23commitments}
    shows that, for every one-way state generator
    $G$, 
    there exists
    a one-way puzzle
    $P = (\Samp, \Ver)$
    such that, if 
    there exists a
    quantum polynomial-time
    algorithm $\cala$
    that breaks 
    $P$,
    then
    there exists a
    quantum polynomial-time
    algorithm $\cala^*$ that breaks
    $G$.
    The corollary then follows by
    plugging the 
    algorithm of 
    \Cref{thm:pp-break-owpuzz} for $P$
    -- 
    a
    polynomial-time quantum algorithm with a $\PP$ oracle that breaks
    $P$ 
    --
    into their reduction.
\end{proof}